\documentclass[11pt]{article}
\usepackage[top=1 in,bottom=1in, left=1 in, right=1 in]{geometry}
\usepackage{hyperref}                                                          % paper

\usepackage{amsmath} % assumes amsmath package installed
\usepackage{amssymb}  % assumes amsmath package installed
\usepackage{bbm}
\usepackage{color}
\usepackage{amsthm}

\DeclareMathOperator{\supp}{supp}
\DeclareMathOperator{\Span}{span}
\DeclareMathOperator{\tr}{tr}

\newtheorem{theorem}{Theorem}[section]
\newtheorem{corollary}{Corollary}[section]
\newtheorem{lemma}{Lemma}[section]
\newtheorem{proposition}{Proposition}[section]

\theoremstyle{definition}
\newtheorem{definition}{Definition}[section]
\newtheorem{remark}{Remark}[section]

\title{\LARGE \bf
State Space Decomposition of Quantum Dynamical Semigroups
}

\author{Nicolas Mousset and Nina H. Amini% <-this % stops a space
\thanks{This work was supported by the ANR projects QuanTEdu-France (ANR-22-CMAS-0001), Q-COAST (ANR-19-CE48-0003) and IGNITION (ANR-21- CE47-0015).}% <-this % stops a space
\thanks{N. Mousset is with CNRS, Laboratoire des Signaux et Syst\`{e}mes (L2S), Universit\'{e} Paris-Saclay,
        CentraleSupélec, 91190 Gif-sur-Yvette, France
        {\tt\small nicolas.mousset@centralesupelec.fr}}%
\thanks{N. H. Amini is with CNRS, Laboratoire des Signaux et Syst\`{e}mes (L2S), Universit\'{e} Paris-Saclay,
        CentraleSupélec, 91190 Gif-sur-Yvette, France
        {\tt\small nina.amini@centralesupelec.fr}}%
}

\begin{document}

\maketitle
\thispagestyle{empty}
\pagestyle{empty}

%%%%%%%%%%%%%%%%%%%%%%%%%%%%%%%%%%%%%%%%%%%%%%%%%%%%%%%%%%%%%%%%%%%%%%%%%%%%%%%%
\begin{abstract}

The mean evolution of an open quantum system in continuous time is described by a time continuous semigroup of quantum channels (completely positive and trace-preserving linear maps). Baumgartner and Narnhofer  presented a general decomposition of the underlying Hilbert space into a sum of invariant subspaces, also called enclosures. We propose a new reading of this result, inspired by the work of Carbone and Pautrat. In addition, we apply this decomposition to a class of open quantum random walks and to quantum trajectories, where we study its uniqueness.
%In addition, we provide some applications on open quantum random walks and quantum trajectories.

\end{abstract}

%%%%%%%%%%%%%%%%%%%%%%%%%%%%%%%%%%%%%%%%%%%%%%%%%%%%%%%%%%%%%%%%%%%%%%%%%%%%%%%%
\section{INTRODUCTION}

Controlling open quantum systems has received significant attention \cite{belavkin1983theory,doherty2000quantum, doherty1999feedback,wiseman2009quantum} and plays an important role in the development of quantum technologies. To systematically control an open quantum system and efficiently analyze stabilization or engineering of states or subspaces, it is necessary to study its long-term behavior. In this regard, the decomposition of the Hilbert space into its recurrent and transient parts, similar to Markov chains, plays a crucial role.
%For instance, the invariant states of a quantum channel are central to study the ergodic properties of an open quantum system \cite{kummerer2004pathwise,wolf2012quantum}.

An open quantum system is a quantum system \textit{A} interacting with a second external quantum system \textit{B}, e.g. a reservoir or a probe
%\cite{breuer2002theory}
(see e.g., \cite{wiseman2009quantum}). When the relaxation time of \textit{A} is very short compared to the correlation time of \textit{B}, the evolution of \textit{A} without any measurement is described by the {\it Markovian master equation}:
$$\frac{d \rho_t}{dt} = \mathcal{L} \left( \rho_t \right),$$
where $\rho_t$ is a density matrix describing the state of system \textit{A}, i.e., a positive matrix with trace one, and $\mathcal{L}$ is the {\it Liouvillian} or {\it Lindladian}. Let $\mathcal{H}$ be the Hilbert space related to \textit{A}, $\mathcal{B(H)}$ the set of linear operators on $\mathcal{H}$ and $\mathcal{D(H)}$ the set of density matrix in $\mathcal{B(H)}$. This way, $\mathcal{L}$ is a linear superoperator on $\mathcal{B(H)}$. The time evolution of a solution of this equation can be described by a semigroup $\left( \Phi_t \right)_{t\ge0}$ of \textit{completely positive trace-preserving} maps that is assumed \textit{continuous} in time \cite{wolf2012quantum}. The semigroup property means that:
$$\Phi_{t+s} = \Phi_t \circ \Phi_s.$$
Throughout this paper, we assume that $\mathcal{H}$ is \textit{finite-dimensional}. This way, we simply have:
$\Phi_t = e^{t \mathcal{L}}$.
The Lindbladian can be expressed explicitly as a conditional completely positive map (see e.g., \cite[Proposition 7.2]{wolf2012quantum}):
\begin{equation*}
    \mathcal{L} = - \imath [H, \cdot] + \sum_{j=1}^n L_j \cdot L_j^\dag - \frac{1}{2} \left\{ L_j^\dag L_j, \cdot \right\},
\end{equation*}
where $H$ is the Hamiltonian and $L_j$ are operators that can be associated to measuring instruments.

The decomposition of the Hilbert space with respect to the semigroup $\left( \Phi_t \right)_{t\geq 0}$ into invariant subspaces for finite-dimensional Hilbert spaces was fully formalized in \cite{baumgartner2012structures}. Later, in \cite{carbone2016open}, this result was extended to the case of open quantum random walks, which address the infinite-dimensional setting. It was further shown that this approach includes a broader class than that of open quantum random walks \cite{carbone2016irreducible}. Subsequently, these results found applications in demonstrating the selection of invariant subspaces by quantum trajectories \cite{ amini2021asymptotic, amini2024exponential,benoist2024exponentially}.

In this paper, we provide a second reading of the works in \cite{baumgartner2012structures, carbone2016irreducible}. Although this theory is well known, we present it again to help clarify these fundamental results. We also detail some proofs, either to simplify them or to precise some important intermediate steps. It should be noted that there are two ways of expressing the decomposition: one that is described in \cite{baumgartner2012structures}, see Theorem 7, but also in \cite{wolf2012quantum}, see Theorem 6.14, and another one described in \cite{carbone2016irreducible}, see Theorem 7.1. We present here an adaptation of \cite{carbone2016irreducible} to the finite-dimensional and continuous-time case, adaptation based largely on \cite{baumgartner2012structures}. In this way, we close the loop opened by these two papers. In the second part of this paper, we apply this theory to a special class of \textit{open quantum random walks} (\textit{OQRW}). Finally, we show that the \textit{identifiability assumption} used to establish the selection of minimal invariant subspaces of quantum trajectories (see \cite{amini2021asymptotic, amini2024exponential,benoist2014large}) provides a sufficient condition for the uniqueness of the decomposition.

%%%%%%%%%%%%%%%%%%%%%%%%%%%%%%%%%%%%%%%%%%%%%%%%%%%%%%%%%%%%%%%%%%%%%%%%%%%%%%%%
\section{Preliminary notions}

We begin with some definitions that will be central for the decomposition of the Hilbert space.

\subsection{Enclosures and invariant states}

\textit{Enclosures}, also called \textit{invariant subspaces}, are subspaces from which it is impossible to get out. 

\begin{definition} [\cite{baumgartner2012structures}]
    A subspace $\mathcal{V} \subseteq \mathcal{H}$ is an {\it enclosure} if for any density matrix $\rho \in \mathcal{D(H)}$ with support in $\mathcal{V}$, for all $t \in \mathbb{R}_+$, $\Phi_t (\rho)$ has also support in $\mathcal{V}$.
\end{definition}

A particular case of enclosures is of special interest for us: the \textit{minimal enclosures} or \textit{minimal invariant subspaces}.

\begin{definition} [\cite{baumgartner2012structures}]
    A non-trivial enclosure is a \textit{minimal enclosure} if the only non-trivial enclosure it contains is itself.
\end{definition}

% \begin{remark}
%     Note that the sum of two enclosures is still an enclosure, see \cite{carbone2016irreducible}, Corollary 4.3.
% \end{remark}

The notion of enclosure is deeply related to the notion of \textit{invariant state}, also called \textit{steady state}, of the semigroup.

\begin{definition}
    A state $\rho \in \mathcal{D(H)}$ is an \textit{invariant state} if for all $t \in \mathbb{R}_+$, $\Phi_t (\rho) = \rho$.
\end{definition}

% \begin{proposition}
%     Equivalently, a state $\rho \in \mathcal{D(H)}$ is invariant iff $\mathcal{L}(\rho) = 0$. 
% \end{proposition}

% \begin{proof}
%     If $\mathcal{L}(\rho) = 0$, then for any $t\ge 0$, $\Phi_t (\rho) = e^{t\mathcal{L}} (\rho) = \rho$. On the other hand, if for all $t\ge0$, $\Phi_t (\rho) = \rho$, since
%     \begin{equation*}
%         \mathcal{L} (\rho) = \frac{\Phi_t (\rho) - \rho}{t} - \sum_{n\ge2} \frac{t^{n-1}}{n!} \mathcal{L}^n (\rho)
%     \end{equation*}
%     and since $\lim_{t \to 0} \sum_{n\ge2} \frac{t^{n-1}}{n!} \mathcal{L}^n (\rho) = 0$, $\mathcal{L}$ being a bounded superoperator, necessarily $\mathcal{L}(\rho) = 0$. \hfill $\square$
% \end{proof}

% \vspace{2mm}

It should be noted that in finite dimension a quantum dynamical semigroup always has at least one invariant state \cite{wolf2012quantum}. Exactly as before, a particular case of invariant states will be of special interest for us: the \textit{extremal invariant states}.

\begin{definition}
    An invariant state is \textit{extremal} if it is an extremal point of the convex set of invariant states.
    % \begin{equation*}
    %     \left\{ \Phi_t (\rho) = \rho, \ \forall t \in \mathbb{R}_+; \ \rho \in \mathcal{D(H)} \right\}
    % \end{equation*}
\end{definition}

\subsection{Recurrent and transient subspaces}

The first stage of the decomposition of the Hilbert space $\mathcal{H}$ consists in the separation into two parts: one that will always be reach as $t$ goes to infinity whatever the initial density matrix and its orthogonal.

\begin{definition} [\cite{carbone2016irreducible}]
    The \textit{recurrent subspace} $\mathcal{R}$ is defined as follows:
    \begin{equation*}
        \mathcal{R} := \sup \left\{ \supp(\rho), \ \rho \in \mathcal{D}(\mathcal{H}), \ \rho \ \mathrm{invariant} \right\}.
    \end{equation*}
    Then, the \textit{transient subspace} $\mathcal{D}$ is defined as: $\mathcal{D} = \mathcal{R}^\perp$.
\end{definition}

% \begin{remark}
%     In finite dimension, an equivalent definition of $\mathcal{D}$ that can be found in \cite{baumgartner2012structures} is:
%     \begin{equation*}
%         \mathcal{D} := \left\{ | \psi \rangle \in \mathcal{H}; \ \lim_{t \to \infty} \langle \psi | \Phi_t (\rho) | \psi \rangle = 0, \ \forall \rho \in \mathcal{D(H)} \right\}
%     \end{equation*}
% \end{remark}

%%%%%%%%%%%%%%%%%%%%%%%%%%%%%%%%%%%%%%%%%%%%%%%%%%%%%%%%%%%%%%%%%%%%%%%%%%%%%%%%%%%%%%%%%%%%%%%%%%%

\section{Decomposition of the Hilbert space}

The objective of the decomposition of the Hilbert space associated to a quantum dynamical semigroup is to have a complete knowledge of the behavior in large time of the mathematical and physical notions related to this semigroup. In this section, we focus on the theoretical propositions while the next sections will present some applications. 

This section is divided into two parts: the proof of the existence of a decomposition and the characterization of the uniqueness of this decomposition.

\subsection{Orthogonal decomposition of $\mathcal{R}$}

By definition, the Hilbert space can be decomposed into two orthogonal parts: $\mathcal{H} = \mathcal{D} \oplus \mathcal{R}$. We can specify the form of $\mathcal{R}$ using the notion of minimal enclosure.

\begin{theorem} [\cite{baumgartner2012structures}] \label{Thm_First_Decomp_R}
    The subspace $\mathcal{R}$ can be decomposed into a direct sum of mutually orthogonal subspaces, where each one of them is a minimal enclosure.
\end{theorem}

The proof of this Theorem is done in two steps in \cite{baumgartner2012structures}: first, they show that $\mathcal{R}$ is an enclosure, second, they prove that the semigroup maps "block by block". For the sake of completeness, we recall the main steps. We restate the first step with a new proof.
\begin{proposition} [\cite{baumgartner2012structures}]
    The support of an invariant state is an enclosure. In particular, $\mathcal{R}$ is an enclosure.
\end{proposition}

\begin{proof}
    We only prove that $\mathcal{R}$ is an enclosure. Let $\rho_0$ be an invariant state. Either all the other invariant states have their support included in the support of $\rho_0$, or there is a $\rho_1$ for which it is not the case. Note that $\left(\rho_0 + \rho_1 \right)/2$ is also an invariant state, with a bigger support. By repeating this operation a finite number of times, we obtain an invariant state whose support is exactly $\mathcal R$, by definition of the recurrent subspace. This way, $\mathcal{R}$ is an enclosure.
\end{proof}

\vspace{2mm}

The next proposition summarizes the second step of the proof. We confuse $\mathcal{H}$ with $\mathbb{C}^n$ and $\mathcal{B(H)}$ with $\mathcal{M}_n(\mathbb{C})$.

\begin{proposition} [\cite{baumgartner2012structures}] \label{Prop_Phi_Block_by_Block}
    Consider an enclosure $\mathcal{V} \subseteq \mathcal{R}$. Then $\mathcal{W} := \mathcal{V}^\perp \cap \mathcal{R}$ is also an enclosure. Moreover, when restricted to $\mathcal{R}$, $\left( \Phi_t \right)_{t\ge0}$ acts as follows. Let $\rho \in\mathcal{D(H)}$ be a state with support in $\mathcal{R}$, then:
    \begin{equation*}
        \Phi_t \left( \begin{pmatrix} \rho_{\mathcal{V},\mathcal{V}} & \rho_{\mathcal{V},\mathcal{W}} \\ \rho_{\mathcal{W},\mathcal{V}} & \rho_{\mathcal{W},\mathcal{W}} \end{pmatrix} \right) = \begin{pmatrix} \Phi_t \left( \rho_{\mathcal{V},\mathcal{V}} \right) & \Phi_t \left( \rho_{\mathcal{V},\mathcal{W}} \right) \\ \Phi_t \left( \rho_{\mathcal{W},\mathcal{V}} \right) & \Phi_t \left(  \rho_{\mathcal{W},\mathcal{W}} \right) \end{pmatrix},
    \end{equation*}
    where, for example, $\rho_{\mathcal{V},\mathcal{W}} = P_\mathcal{V} \rho P_\mathcal{W}$, $P_\mathcal{V}$, resp. $P_\mathcal{W}$, being the Hermitian projector onto $\mathcal{V}$, resp. $\mathcal{W}$.
\end{proposition}

\paragraph*{Proof of Theorem \ref{Thm_First_Decomp_R}}
    Given a minimal enclosure $\mathcal{V}_1 \subseteq \mathcal{R}$, the subspace $\mathcal{X}_2 = \mathcal{V}_1^\perp \cap \mathcal{R}$ is also an enclosure. Either it is minimal, which finishes the proof, or it contains a minimal enclosure $\mathcal{V}_2$. Repeating this stage with $\mathcal{V}_2$ and so on and so forth finishes the proof. \hfill $\square$

\vspace{2mm}

Finally, note that these results can also be interpreted in terms of invariant states. %The next proposition corresponds to Proposition 15 in \cite{baumgartner2012structures}.

\begin{proposition} [\cite{baumgartner2012structures}] \label{Prop_Inv_States_Support}
    \begin{enumerate}
        \item Each enclosure inside of $\mathcal{R}$ is the support of an invariant state.
        \item Different extremal invariant states have different supports.
        \item The support of an extremal invariant state is a minimal enclosure inside of $\mathcal{R}$ and vice-versa.
    \end{enumerate}
\end{proposition}

To summarize, the recurrent subspace can always be decomposed into an orthogonal sum of minimal enclosures. To study the uniqueness of this decomposition, we introduce the \textit{cut off evolution semigroup}.

\subsection{The cut off evolution semigroup}

The cut off evolution semigroup $\left( S_t \right)_{t\ge0}$ corresponds to the adjoint semigroup $\left( \Phi_t^* \right)_{t\ge0}$ restricted to the recurrent subspace.

\begin{definition} [\cite{baumgartner2012structures}]
    Given an operator $A \in \mathcal{B(H)}$ with support in $\mathcal{R}$, for all $t \ge 0$, $S_t$ is defined as:
    \begin{equation*}
        \mathcal{S}_t(A) := P_\mathcal{R} \Phi_t^* (A) P_\mathcal{R}.
    \end{equation*}
\end{definition}

\begin{remark}
    As its name suggests, the cut off evolution semigroup is a semigroup, see \cite{baumgartner2012structures}.
\end{remark}

As for the original semigroup, the cut off evolution semigroup maps block by block:

\begin{proposition} [\cite{baumgartner2012structures}]
    Consider two enclosures $\mathcal{V} \subseteq \mathcal{R}$ and $\mathcal{W} = \mathcal{V}^\perp \cap \mathcal{R}$. Then, for any operator $A \in \mathcal{B(H)}$ and any $t \ge 0$:
    \begin{equation*}
        S_t \left( \begin{pmatrix} A_{\mathcal{V},\mathcal{V}} & A_{\mathcal{V},\mathcal{W}} \\ A_{\mathcal{W},\mathcal{V}} & A_{\mathcal{W},\mathcal{W}} \end{pmatrix} \right) = \begin{pmatrix} S_t \left( A_{\mathcal{V},\mathcal{V}} \right) & S_t \left( A_{\mathcal{V},\mathcal{W}} \right) \\ S_t \left( A_{\mathcal{W},\mathcal{V}} \right) & S_t \left(  A_{\mathcal{W},\mathcal{W}} \right) \end{pmatrix},
    \end{equation*}
    with, for example, $A_{\mathcal{V},\mathcal{W}} = P_\mathcal{V} A P_\mathcal{W}$.
\end{proposition}

The cut off evolution semigroup has the useful following properties:

\begin{proposition} [\cite{baumgartner2012structures}] \label{Prop_Cut_Off_Inv_Proj}
    \begin{enumerate}
        \item A non-trivial Hermitian projector $P_\mathcal{V}$ is invariant under $\left( \mathcal{S}_t \right)_{t\ge0}$ iff it projects onto an enclosure $\mathcal{V} \subseteq \mathcal{R}$.
        %\item A Hermitian operator with support in $\mathcal{R}$ is invariant under $\left( \mathcal{S}_t \right)_{t\ge0}$ iff its spectral projections are invariant.
        \item $\mathcal{V} \subseteq \mathcal{R}$ is a minimal enclosure iff $P_\mathcal{V}$ cannot be decomposed into a sum of invariant Hermitian projectors.
    \end{enumerate}
\end{proposition}

Therefore, it can be used to completely characterize the Hermitian projectors onto enclosures.
%This will prove useful in addressing the question of the uniqueness of the decomposition of the Hilbert space.

\subsection{Uniqueness of the orthogonal decomposition}

%We study the uniqueness of the decomposition in continuous time in the same way as Carbone and Pautrat did in discrete time.
We adapt Carbone and Pautrat's method to study uniqueness in continuous time. To aid understanding, we provide more detailed proofs when necessary due to their importance. 

\begin{proposition}  [\cite{carbone2016irreducible}] \label{Prop_Non_Unique_Decomp}
    Consider two different minimal enclosures $\mathcal{V}_1, \mathcal{V}_2 \subseteq \mathcal{R}$ and let $\mathcal{V} := \mathcal{V}_1 \oplus \mathcal{V}_2$. The decomposition of $\mathcal{V}$ into a direct sum of minimal enclosures is unique iff any minimal enclosure $\mathcal{W}$ not simultaneously orthogonal to both $\mathcal{V}_1$ and $\mathcal{V}_2$ verifies $\mathcal{V} \cap \mathcal{W} = \left\{ 0 \right\}$.
\end{proposition}

\begin{corollary} [\cite{carbone2016irreducible}]
    If $\rho \in \mathcal{D}(\mathcal{H})$ is invariant and if $\mathcal{V}_1, \mathcal{V}_2 \subseteq \mathcal{R}$ are two different minimal enclosures such that the decomposition of $\mathcal{V}_1 \oplus \mathcal{V}_2$ into a sum of minimal enclosures is unique, then:
    \begin{equation*}
        P_{\mathcal{V}_1} \rho P_{\mathcal{V}_2} = P_{\mathcal{V}_2} \rho P_{\mathcal{V}_1} = 0.
    \end{equation*}
\end{corollary}

%Recall that there always exists an orthogonal decomposition. Thus, if the decomposition is unique, the minimal enclosures that appear in it are automatically mutually orthogonal. Furthermore, there is another sufficient condition to characterize uniqueness of the decomposition. It is the contraposition of the following lemma:
Recall that an orthogonal decomposition always exists. Thus, if the decomposition is unique, the minimal enclosures involved must be mutually orthogonal. Moreover, uniqueness can also be characterized by another sufficient condition, given by the contraposition of the following lemma.
\begin{lemma} [\cite{carbone2016irreducible}] \label{Lemma_Equality_Dim}
    If $\mathcal{V}_1, \mathcal{V}_2 \subseteq \mathcal{R}$ are two different minimal enclosures such that the decomposition of $\mathcal{V} = \mathcal{V}_1 \oplus \mathcal{V}_2$ into a sum of minimal enclosures is not unique, then they have the same dimension.
\end{lemma}

\begin{proof}
    According to Proposition \ref{Prop_Non_Unique_Decomp}, there exists a minimal enclosure $\mathcal{W}$ not simultaneously orthogonal to both $\mathcal{V}_1$ and $\mathcal{V}_2$ such that $\mathcal{V} \cap \mathcal{W} \ne \left\{ 0 \right\}$. Let $\rho_\mathcal{W}$, resp. $\rho_{\mathcal{V}_1}$, $\rho_{\mathcal{V}_2}$, be the invariant states associated to $\mathcal{W}$, resp. $\mathcal{V}_1$, $\mathcal{V}_2$ and $P_\mathcal{W}$, resp. $P_{\mathcal{V}_1}$, $P_{\mathcal{V}_2}$ the associated Hermitian projectors. Recall that, by Proposition \ref{Prop_Phi_Block_by_Block}, $P_{\mathcal{V}_1} \mathcal{D(H)} P_{\mathcal{V}_1}$ evolves independently from the rest. Thus, $P_{\mathcal{V}_1} \rho_\mathcal{W} P_{\mathcal{V}_1}$ is an invariant state, up to renormalization, with support in the minimal enclosure $\mathcal{V}_1$. It is then proportional to $\rho_{\mathcal{V}_1}$ by Proposition \ref{Prop_Inv_States_Support}. So, the support of $\rho_{\mathcal{V}_1}$ is included in the support of $\rho_\mathcal{W}$. In particular, $\dim (\mathcal{V}_1) \le \dim (\mathcal{W})$. Doing the same with $P_\mathcal{W} \rho_{\mathcal{V}_1} P_\mathcal{W}$, we have $\dim (\mathcal{V}_1) = \dim (\mathcal{W})$. Applying the same process to $\mathcal{V}_2$ completes the proof.
\end{proof}

\vspace{2mm}

The following proposition introduces a partial isometry that is central to characterize the general form of the decomposition of the Hilbert space.

\begin{proposition} [\cite{carbone2016irreducible}] \label{Prop_Partial_Isometry}
    Let $\mathcal{V}_1, \mathcal{V}_2 \subseteq \mathcal{R}$ be two orthogonal minimal enclosures such that the decomposition of $\mathcal{V}_1 \oplus \mathcal{V}_2$ into a sum of minimal enclosures is not unique. Then there exists a partial isometry $Q \in \mathcal{B}(\mathcal{H})$ from $\mathcal{V}_1$ to $\mathcal{V}_2$ satisfying
    \begin{equation*}
        Q^\dag Q = \mathbbm{1}_{\mathcal{V}_1} \quad and \quad Q Q^\dag = \mathbbm{1}_{\mathcal{V}_2}
    \end{equation*}
    such that, for any $\rho \in \mathcal{D}(\mathcal{H})$ with support in $\mathcal{R}$, we have, for any $t\ge0$ and $j \in \{ 1,2 \}$,
    \begin{equation*} %\label{Eq_Non_Unique_Minimal_Enclosures}
        R \Phi_t (\rho) P_{\mathcal{V}_j} + P_{\mathcal{V}_j} \Phi_t (\rho) R = \Phi_t \left( R \rho P_{\mathcal{V}_j} + P_{\mathcal{V}_j} \rho R \right),
    \end{equation*}
    with $R := Q P_{\mathcal{V}_1} + Q^\dag P_{\mathcal{V}_2}$.
\end{proposition}

\begin{proof}
    Using the same notations as in the previous proof, by Proposition \ref{Prop_Cut_Off_Inv_Proj}, the only invariant Hermitian projector with respect to the cut off semigroup with support in $\mathcal{V}_1$ (resp. $\mathcal{V}_2$, $\mathcal{W}$) is $P_{\mathcal{V}_1}$ (resp. $P_{\mathcal{V}_2}$, $P_\mathcal{W}$). Thus, exactly as in the previous proof, $P_{\mathcal{V}_1} P_\mathcal{W} P_{\mathcal{V}_1}$ (resp. $P_{\mathcal{V}_2} P_\mathcal{W} P_{\mathcal{V}_2}$) is proportional to $P_{\mathcal{V}_1}$ (resp. $P_{\mathcal{V}_2}$). Moreover:
    %, $P_\mathcal{W}$ has to respect two additional constraints:
    \begin{itemize}
        \item $\dim \left( \mathcal{W} \right) = \dim \left( \mathcal{V}_1 \right) = \dim \left( \mathcal{V}_2 \right)$ by Lemma \ref{Lemma_Equality_Dim},
        \item $P_\mathcal{W} = \left( P_\mathcal{W}\right)^\dag$ and $\left( P_\mathcal{W} \right)^2 = P_\mathcal{W}$.
    \end{itemize}
    Therefore, when restricted to $\mathcal{V}$, in a basis respecting the orthogonal decomposition $\mathcal{V} = \mathcal{V}_1 \oplus \mathcal{V}_2$, $P_\mathcal{W}$ is of the form:
    $$P_\mathcal{W} = \begin{pmatrix} \cos^2(\theta_0) \mathbbm{1}_{\mathcal{V}_1} & \sin(\theta_0) \cos(\theta_0) Q^\dag \\ \sin(\theta_0) \cos(\theta_0) Q & \sin^2(\theta_0) \mathbbm{1}_{\mathcal{V}_2} \end{pmatrix},$$
    with $\theta_0 \in [0,\pi[$ a constant. The pattern of the off-diagonal blocks results from the last constraint, which implies in particular $Q Q^\dag = \mathbbm{1}_{\mathcal{V}_2}$ and $Q^\dag Q = \mathbbm{1}_{\mathcal{V}_1}$. We extend $Q$ to $\mathcal{B}(\mathcal{H})$ such that it acts as the zero operator on the complements of $\mathcal{V}_1$.

    This pattern of projector is in fact very general and gives birth to an infinite number of invariant Hermitian projectors. Indeed, since for any $t\ge0$, $S_t \left( P_\mathcal{W} \right) = P_\mathcal{W}$, and since
    \begin{equation*} \begin{split}
        \mathcal{S}_t \left( P_\mathcal{W} \right) &= \cos^2(\theta_0) P_{\mathcal{V}_1} + \sin^2(\theta_0) P_{\mathcal{V}_2}\\& \phantom{mmm} + \sin(\theta_0) \cos(\theta_0) \begin{pmatrix} 0 & S_t \left( Q^\dag \right) \\ S_t(Q) & 0
        \end{pmatrix},
    \end{split} \end{equation*}
    we can directly deduce that
    $S_t (Q) = Q$ and $S_t \left( Q^\dag \right) = Q^\dag$.
    
    This shows that for, any $\theta \in [0,\pi[$, a Hermitian projector $P_\theta$ of the form
    \begin{equation*}
        P_\theta = \begin{pmatrix}
            \cos^2(\theta) \mathbbm{1}_{\mathcal{V}_1} & \sin(\theta) \cos(\theta) Q^\dag \\ \sin(\theta) \cos(\theta) Q & \sin^2(\theta) \mathbbm{1}_{\mathcal{V}_2}
        \end{pmatrix}
    \end{equation*}
    is invariant, i.e., satisfies, for any $\rho \in \mathcal{D}(\mathcal{H})$ with support in $\mathcal{R}$, for any $t\ge0$,
    \begin{equation*} %\label{Eq_invariant_P_Theta}
        \Phi_t \left( P_\theta \rho P_\theta \right) =  P_\theta \Phi_t (\rho) P_\theta
    \end{equation*}
    because it projects onto an enclosure by Proposition \ref{Prop_Cut_Off_Inv_Proj}. By differentiating this equation, we get:
    \begin{equation*}
        \Phi_t \left( \frac{d P_\theta}{d\theta} \rho P_\theta + P_\theta \rho \frac{d P_\theta}{d\theta} \right) = \frac{d P_\theta}{d\theta} \Phi_t (\rho) P_\theta + P_\theta \rho \frac{d P_\theta}{d\theta}.
    \end{equation*}
    To show the equation in the theorem, it is sufficient to evaluate the above equation in $\theta = 0$ and $\theta = \pi / 2$.
\end{proof}

\vspace{2mm}

% \begin{remark}
%      In the proof, we constructed the whole family of enclosures contained in $\mathcal{V}_1 \oplus \mathcal{V}_2$ when this orthogonal decomposition is not unique: they are fully characterized by the Hermitian projectors $\left( P_\theta \right)_{\theta\in [0,\pi[}$. 
% \end{remark}

The partial isometry $Q$ can be used to describe exactly the relation between the extremal invariant states $\rho_{\mathcal{V}_1}$ and $\rho_{\mathcal{V}_2}$.

\begin{corollary} [\cite{carbone2016irreducible}] \label{Cor_Non_Unique_Minimal_Enclosures}
    With the same assumptions as in Proposition \ref{Prop_Partial_Isometry}, let $\rho_{\mathcal{V}_1}$, resp. $\rho_{\mathcal{V}_2}$, be the extremal invariant state associated to $\mathcal{V}_1$, resp. $\mathcal{V}_2$. Consider $Q$ the partial isometry defined previously. Then
    \begin{equation*} %\label{Eq_Link_Rho_2_and_Rho_1}
        \rho_{\mathcal{V}_2} = Q \rho_{\mathcal{V}_1} Q^\dag.
    \end{equation*}
    In addition, if $\rho_{inv} \in \mathcal{D}(\mathcal{H})$ is an invariant state, then
    \begin{itemize}
        \item $P_{\mathcal{V}_1} \rho_{inv} P_{\mathcal{V}_1}$ is proportional to $\rho_{\mathcal{V}_1}$,
        \item $P_{\mathcal{V}_2} \rho_{inv} P_{\mathcal{V}_2}$ is proportional to $\rho_{\mathcal{V}_2}$,
        \item $P_{\mathcal{V}_1} \rho_{inv} P_{\mathcal{V}_2} Q$ is proportional to $\rho_{\mathcal{V}_1}$,
        \item $P_{\mathcal{V}_2} \rho_{inv} P_{\mathcal{V}_1} Q^\dag$ is proportional to $\rho_{\mathcal{V}_2}$.
    \end{itemize}
\end{corollary}

\begin{proof}
    The proof of the first part can be done in two steps. First, apply the equation in Proposition \ref{Prop_Partial_Isometry} with $j=1$ to $\rho_{\mathcal{V}_1}$. After computations, we find that
    \begin{equation*}
        \rho' := \begin{pmatrix} 0 & \rho_{\mathcal{V}_1} Q^\dag \\ Q \rho_{\mathcal{V}_1} & 0 \end{pmatrix}
    \end{equation*}
    is invariant. Applying again the equation in Proposition \ref{Prop_Partial_Isometry} to $\rho'$ with $j=2$, this time, we have that
    \begin{equation*}
        \rho'' := \begin{pmatrix} 0 & 0 \\ 0 & Q \rho_{\mathcal{V}_1} Q^\dag \end{pmatrix}
    \end{equation*}
    is an invariant state. By uniqueness, it must be equal to $\rho_{\mathcal{V}_2}$.

    For the second part, it has to be shown that, given an invariant state $\rho_{inv}$, $P_{\mathcal{V}_1} \rho_{inv} P_{\mathcal{V}_2} Q$ is proportional to $\rho_{\mathcal{V}_1}$. First, note that the equation in Proposition \ref{Prop_Partial_Isometry} can be extended to any matrix with support in $\mathcal{R}$. Thus, applying it with $j=1$ to $P_{\mathcal{V}_1} \rho_{inv} P_{\mathcal{V}_2}$ results in the fact that
    \begin{equation*}
        \Tilde{\rho} := \begin{pmatrix} P_{\mathcal{V}_1} \rho_{inv} P_{\mathcal{V}_2} Q & 0 \\ 0 & 0 \end{pmatrix}
    \end{equation*}
    is an invariant matrix. Since $\left( \left( \Phi_t \right)_{|\mathcal{V}_1} \right)_{t\ge0}$ is irreducible, see \cite{wolf2012quantum}, the span of its eigenoperators is generated by $\rho_{\mathcal{V}_1}$, hence the proportionality. The last relation can be found in the same way by applying the equation in Proposition \ref{Prop_Partial_Isometry} with $j=2$ to $P_{\mathcal{V}_2} \rho_{inv} P_{\mathcal{V}_1}$.
\end{proof}

\vspace{2mm}

Everything is summarized in the following theorem and corollary.

\begin{theorem} [\cite{carbone2016irreducible}] \label{Thm_General_Decomp}
    There exists a decomposition of $\mathcal{H}$ of the form:
    \begin{equation*} %\label{Eq_Decomp_Hilbert_Space}
        \mathcal{H} = \mathcal{D} \oplus \sum_{\alpha\in A} \mathcal{V}_a \oplus \sum_{\beta \in B} \sum_{\gamma \in C_\beta} \mathcal{V}_{\beta, \gamma},
    \end{equation*}
    where $A$ or $B$ can be empty, $C_\beta$ contains at least two elements, and:
    \begin{itemize}
        \item All the $\mathcal{V}_\alpha$ and $\mathcal{V}_{\beta, \gamma}$ are mutually orthogonal minimal enclosures,
        \item Any minimal enclosure that is not orthogonal to $\sum_{\gamma \in C_\beta} \mathcal{V}_{\beta, \gamma}$ is contained in this sum.
    \end{itemize}
\end{theorem}

\begin{corollary} [\cite{carbone2016irreducible}]
    Let $\rho_{inv} \in \mathcal{D}(\mathcal{H})$ be an invariant state. Consider an orthogonal decomposition of $\mathcal{H}$ as in Theorem \ref{Thm_General_Decomp}. Then:
    \begin{itemize}
        \item $\rho_{inv}$ has support in $\mathcal{R}$.
        \item $P_{\mathcal{V}_\alpha} \rho_{inv} P_{\mathcal{V}_\alpha}$ is proportional to $\rho_{\mathcal{V}_\alpha}$. Same for $\mathcal{V}_{\beta,\gamma}$.
        \item For each $\beta$ and each couple $\gamma \ne \gamma' \in C_\beta$, there exist a partial isometry $Q_{\beta,\gamma,\gamma'}$ as described in Corollary \ref{Cor_Non_Unique_Minimal_Enclosures}.
        \item All the other off-diagonal blocks of $\rho_{inv}$ are zero.
    \end{itemize}
\end{corollary}

%%%%%%%%%%%%%%%%%%%%%%%%%%%%%%%%%%%%%%%%%%%%%%%%%%%%%%%%%%%%%%%%%%%%%%%%%%%%%%%%

\section{Application in dimension 2} \label{Section_App_Dim_2}

To better understand what are the different possible decompositions of a Hilbert space, let us look at the general case when $\mathcal{H} = \mathbb{C}^2$. Recall that the semigroup $\left( \Phi_t \right)_{t\ge0}$ has at least one invariant state and that invariant states are states $\rho$ that satisfies $\mathcal{L}(\rho) = 0$.

\subsection{Only one extremal invariant state}

Assume that the semigroup has exactly one invariant state, which will necessarily be extremal. Then, there are two possible options:

\subsubsection{Faithful invariant state}

It is the case for example if we choose
$H = 0, \ L_1 = |e_1\rangle \langle e_2 | \ \mathrm{and} \ L_2 = |e_2\rangle \langle e_1 |$.
Indeed, for $\rho = \begin{pmatrix} a & b \\ \overline{b} & c \end{pmatrix}$ a density matrix, we have the following Lindbladian: 
$\mathcal{L} (\rho) = \begin{pmatrix} c-a & -b \\ - \overline{b} & a-c \end{pmatrix}$.
Therefore, $\mathcal{L}(\rho)=0$ iff $\rho = \frac{1}{2} \mathbbm{1}_{\mathbb{C}^2}$. In particular, $\mathcal{R} = \mathbb{C}^2$ and $\mathcal{D} = {0}$.

\subsubsection{Unfaithful invariant state}

On the contrary, if we drop $L_2$, we have
$\mathcal{L} (\rho) = \begin{pmatrix} c & -b/2 \\ - \overline{b}/2 & -c \end{pmatrix}$.
Thus, $\mathcal{L}(\rho) = 0$ iff $\rho = | e_1 \rangle \langle e_1 |$. The extremal invariant state is not faithful, wich leads to the following decomposition : $\mathcal{R} = \Span(|e_1\rangle)$ and $\mathcal{D} = \Span(|e_2\rangle)$.

\subsection{Exactly two extremal invariant states}

The semigroup can also have exactly two extremal invariant states. For instance, if we take:
$H = 0 \ \mathrm{and} \ L= |e_1\rangle \langle e_1 |$,
we have
$\mathcal{L}(\rho) = \begin{pmatrix} 0 & - b/2 \\ - \overline{b}/2 & 0 \end{pmatrix}$.
Therefore, $\mathcal{L}(\rho) = 0$ iff $\rho = \rho_1 = | e_1 \rangle \langle e_1 |$, $\rho = \rho_2 = | e_2 \rangle \langle e_2 |$ or $\rho$ is a linear combination of $\rho_1$ and $\rho_2$. This way, we have exactly two minimal enclosures: $\mathcal{V}_1 = \Span(|e_1\rangle)$ and $\mathcal{V}_2 = \Span(|e_2\rangle)$.

\subsection{At least three extremal invariant states}

To consider this last case, we work in the basis of $\mathcal{M}_2(\mathbb{C})$ formed by Pauli's matrices: $\left( \mathbbm{1}_{\mathbb{C}^2}, \sigma_X, \sigma_Y, \sigma_Z \right)$. Since for any $M \in \mathcal{M}_2 (\mathbb{C})$, $\mathcal{L}(M)$ must have trace equals to zero, the matrix form of $\mathcal{L}$ in this basis is necessarily: 
$$\mathcal{L} = \begin{pmatrix} 0 & 0 \\ a & A \end{pmatrix},$$
with $a \in \mathbb{C}^3$ and $A \in \mathcal{M}_3(\mathbb{C})$. Recall that, according to the Bloch sphere formalism, any density matrix can be written as $\rho = \frac{1}{2} \mathbbm{1}_{\mathbb{C}^2} + u \cdot \sigma$ with $u \in \mathbb{R}^3, ||u|| \le 1$ and $\sigma = \left( \sigma_X, \sigma_Y, \sigma_Z \right)$. So, $\mathcal{L}(\rho) = 0$ is equivalent to $Au = - a/2$.
%$$\frac{a}{2} + Au = 0$$

Let's assume that there are at least three extremal invariant states. This means that there is a basis of $\mathbb{R}^3$, $\left( u_1, u_2, u_3 \right)$, that is solution of this system. Thus, any vector $u\in \mathbb{R}^3$ is solution of the system, so any density matrix is invariant. In other words, $\mathcal{L}$ is the null superoperator.

Note that in this case, the family of Hermitian projectors $\left( P_\theta \right)_{\theta \in [0,\pi[}$ described in the proof of Proposition \ref{Prop_Partial_Isometry} corresponds to the family of Hermitian projectors on one-dimensional subspaces of $\mathbb{C}^2$. This shows that any one-dimensional subspace is a minimal enclosure, which is again possible iff $\mathcal{L}$ is the null superoperator.

%%%%%%%%%%%%%%%%%%%%%%%%%%%%%%%%%%%%%%%%%%%%%%%%%%%%%%%%%%%%%%%%%%%%%%%%%%%%%%%%%%%%%%%

\section{Application to Open quantum random walks}

The model of \textit{open quantum random walks} has been introduced in \cite{attal2012open} by Attal, Pettruccione, Sabot and Sinayskiy. The starting point is an oriented graph $\left\{ (i,j); \ i,j \in \mathcal{G} \right\}$, which we will consider finite, and two Hilbert spaces $\mathcal{H}$ and $\mathcal{K}$. The latter is chosen so that its dimension is equal to the number of vertices of $\mathcal{G}$. Each edge $(i,j)$ is associated with an operator $B^i_j \in \mathcal{B(H)}$ such that:
\begin{equation*}
    \forall j, \hspace{2mm} \sum_i \left( B^i_j \right)^\dag B^i_j = \mathbbm{1}_\mathcal{H}.
\end{equation*}
Then, define the following bounded operators on $\mathcal{H} \otimes \mathcal{K}$:
\begin{equation*}
    V^i_j = B^i_j \otimes | i \rangle \langle j |,
\end{equation*}
where $( | i \rangle )_{i}$ is an orthonormal basis of $\mathcal{K}$. Thus, 
\begin{equation*}
    \Phi := \sum_{i,j} V^i_j \cdot \left( V^i_j \right)^\dag
\end{equation*}
is a quantum channel on $\mathcal{H} \otimes \mathcal{K}$. The random walk is the following: if the system state at time $n$ is $\rho^{(n)} \otimes | i \rangle \langle i |$, then at time $n+1$ it will become $B^i_j \rho^{(n)} \left( B^i_j \right)^\dag \otimes | j \rangle \langle j |$ with probability $\tr \left( B^i_j \rho^{(n)} \left( B^i_j \right)^\dag \right)$. Thus, we are performing a random walk in $\mathcal{K}$ dictated by the evolution in $\mathcal{H}$.

In \cite{pellegrini2014continuous}, Pellegrini applied the same idea in continuous time, leading to continuous-time open quantum random walks. A semigroup of quantum channels $\left( \Phi_t \right)_{t\ge0}$ appears, generated by the usual Lindbladian on $\mathcal{H} \otimes \mathcal{K}$:
\begin{equation*}
    \mathcal{L} = - \imath [H, \cdot] + \sum_{i,j} L^i_j \cdot \left( L^i_j \right)^\dag - \frac{1}{2} \left\{ \left( L^i_j \right)^\dag L^i_j, \cdot \right\},
\end{equation*}
where $L^i_j = B^i_j \otimes |i \rangle \langle j | \hspace{2mm} \mathrm{and} \hspace{2mm} H = \sum_i H_i \otimes | i \rangle \langle i |$.

We propose to study the decomposition of the Hilbert space $\mathcal{H}$ for a specific class of open quantum random walks: those equivalent to continuous-time Markov chains. According to \cite{carbone2016open}, we call them \textit{minimal OQRW} in the sense that we restrict the first Hilbert space $\mathcal{H}$ to $\mathbb{C}$. This way, $\mathcal{H} \otimes \mathcal{K}$ is isomorphic to $\mathbb{C}^n$.

A continuous-time Markov chain in a state space of $n$ elements can be represented by a rate matrix $Q = \left( q_{i,j} \right)_{1 \le i,j \le n}$ such that:
$$\forall i \ne j, q_{i,j} \ge 0,$$
$$\forall i, \sum_j q_{i,j} = 0, \ {\it i.e.}, \ q_{i,i} = -\sum_{j\ne i} q_{i,j} \le 0.$$
According to \cite{pellegrini2014continuous}, to encode this process in an open quantum random walk, we can choose:
\begin{equation} \label{Eq_minimal_OQRW}
    H = 0, \hspace{2mm} \forall i \ge 0, \hspace{2mm} B_i^i=0, \hspace{2mm} \mathrm{and} \hspace{2mm} \forall i\ne j, \ B^i_j = \sqrt{q_{i,j}}.
\end{equation}

\begin{theorem}
    The invariant states of a minimal OQRW described by Equation (\ref{Eq_minimal_OQRW}) are the invariant measures of the continuous-time Markov chain with rate matrix Q. Its minimal enclosures are the closed communication classes of the Markov chain.
\end{theorem}

\begin{proof}
    The Lindbladian is as follows:
    \begin{equation*} \begin{split}
        \mathcal{L} (\rho) &= \sum_{i \ne j} q_{i,j} | j \rangle \langle i | \rho | i \rangle \langle j | \\ & \hspace{10mm} - \frac{1}{2} \sum_i ( \sum_{j \ne i} q_{i,j} ) \left( | i \rangle \langle i | \rho + \rho | i \rangle \langle i | \right)
.    \end{split} \end{equation*}
    By definition, $\sum_{j \ne i} q_{i,j} = - q_{i,i}$. Moreover, using that $\rho = \sum_{k,l} \langle k | \rho | l \rangle | k \rangle \langle l |$, we obtain the following expression:
    \begin{equation*} \begin{split}
        \mathcal{L}(\rho) &= \sum_j \left( \sum_i q_{i,j} \langle i | \rho | i \rangle \right) | j \rangle \langle j | \\
        & \hspace{10mm} + \frac{1}{2} \sum_{k \ne j} \left( q_{j,j} + q_{k,k} \right) \langle j | \rho | k \rangle |j \rangle \langle k |.
    \end{split} \end{equation*}
    Therefore, $\mathcal{L}(\rho) = 0$ if and only if:
    \begin{itemize}
        \item $\forall j, \ \sum_i q_{i,j} \langle i | \rho | i \rangle = 0$,
        \item $\forall j \ne k, \ \left( q_{j,j} + q_{k,k} \right) \langle j | \rho | k \rangle = 0$.
    \end{itemize}
    Assuming that the state space is chosen wisely, for all $j$, $q_{j,j} \ne 0$. Therefore, an invariant state will necessarily be diagonal. 
    % If we denote $\pi_i = \langle i | \rho | i \rangle$, the first equation becomes:
    % $$\forall j, \ \sum_i q_{i,j} \pi_i = 0$$
    On the other hand, if $\pi_i = \langle i | \rho | i \rangle$, the first equation becomes $\sum_i q_{i,j} \pi_i = 0$ for all $j$, which is exactly the definition of an invariant measure of $Q$, see \cite{anderson2012continuous}. Consequently, the extremal invariant states are exactly the diagonal states with an extremal invariant measure on their diagonal.
    
    Recall that the support of the extremal invariant measures of a continuous-time Markov chain are the closed communication classes, a closed communication class being a part of the state space in which the process stay stuck forever once reached and where each component can be attained with positive probability. As a result, the Hilbert space can be decomposed into an orthogonal sum of communicating classes. In this particular case, the notions of minimal enclosures and of communication classes come together.
\end{proof}

\vspace{2mm}

Another situation where these two notions are connected, though in a trivial way, is for quantum trajectories under both a purification and an ergodicity assumption, see \cite{benoist2021invariant}. Under these assumptions, a quantum trajectory has only one invariant measure and only one minimal enclosure. 

The next section takes a closer look at quantum trajectories and considers sufficient conditions to ensure the uniqueness of the decomposition of the associated Hilbert space.

%%%%%%%%%%%%%%%%%%%%%%%%%%%%%%%%%%%%%%%%%%%%%%%%%%%%%%%%%%%%%%%%%%%%%%%%%%%%%%%%%%%%%%%%%%%%

\section{Application to quantum trajectories}

A \textit{quantum trajectory} models the time evolution of an open quantum system undergoing repeated indirect measurements in discrete or continuous time. Its mean evolution, understood without measurement, is described by a quantum dynamical semigroup. Therefore, we can decompose the underlying Hilbert space with respect to this semigroup, as presented earlier. This sheds light on the mean behavior in large time of the quantum trajectory.  What is remarkable is that even with indirect measurement, the system's state can still converge towards one of the minimal enclosures at large time, see for example \cite{amini2021asymptotic, amini2024exponential, benoist2014large}. In each paper, an assumption is made to ensure the selection of a minimal enclosure, known as \textit{non degeneracy} or \textit{identifiability}. However, it has never been verified that, under these assumptions, the decomposition of the Hilbert space is unique. This should be a natural result, as it would be strange to select a minimal enclosure that could vanish at any time or merge with another. We propose to verify the uniqueness of the decomposition under these assumptions.

\subsection{Quantum nondemolition measurement in continuous time}

Let us begin our study with a quantum trajectory in continuous time undergoing a \textit{nondemolition measurement}, see \cite{benoist2014large} for details. According to Theorem 2 in \cite{benoist2014large}, the nondemolition condition is equivalent to saying that all operators $H$ and $L_j$ appearing in the Lindbladian are diagonal in the same basis $\left( | \alpha \rangle \right)_{0 \le \alpha \le n}$. In details, we write:
$$H = \sum_\alpha \epsilon(\alpha) | \alpha \rangle \langle \alpha |, \ \epsilon(\alpha) \in \mathbb{R},$$
$$L_j = \sum_{\alpha} c(j|\alpha) | \alpha \rangle \langle \alpha |, \ c(j|\alpha) \in \mathbb{C}. $$
The fact that $H$ and the $L_j$ are diagonal implies that the pointer states $| \alpha \rangle \langle \alpha |$ are extremal invariant states. In this way, we can directly infer a first decomposition of the Hilbert space:
$$\mathcal{H} = \bigoplus_\alpha \Span (|\alpha \rangle).$$
While this decomposition is always used in practice, there is no reason for it to be the only possible one. Indeed, to show selection of a pointer state by the quantum trajectory as time goes to infinity, a second assumption is required: the \textit{non degeneracy assumption}. This assumption is expressed in two parts since there exists an integer $0 \le p \le n$ such that for $0 \le j \le p$, $L_j$ is associated to a Wiener process, while for $p+1 \le j \le n$, $L_j$ is associated to a Poisson process, see the \textit{stochastic master equation} in \cite{benoist2014large}. Denoting $r(j|\alpha) = c(j|\alpha) + \overline{c(j|\alpha)}$ and $\theta(j|\alpha) = |c(j|\alpha)|^2$, we assume that for all $\alpha \ne \beta$:
\begin{itemize}
    \item either there exists $0 \le j \le p$ such that $r(j|\alpha) \ne r(j|\beta)$,
    \item or there exists $p+1 \le j \le n$ such that $\theta(j|\alpha) \ne \theta(j|\beta)$.
\end{itemize}

\begin{theorem}
    The decomposition of the Hilbert space associated with a quantum trajectory undergoing nondemolition measurements, under the non degeneracy assumption, is unique.
\end{theorem}

\begin{proof}
    To prove the uniqueness of the decomposition, as usual, we have to look at the invariant states of the Lindbladian. To begin with, let's express it explicitly. For a density matrix $\rho$, we have:
    %$$H \rho - \rho H = \sum_{\alpha, \beta} \langle \alpha | \rho | \beta \rangle \left( \epsilon (\alpha) - \epsilon (\beta) \right) | \alpha \rangle \langle \beta |$$
    \begin{equation*} \begin{split}
        H \rho - \rho H = \sum_{\alpha, \beta} \langle \alpha | \rho | \beta \rangle \left( \epsilon (\alpha) - \epsilon (\beta) \right) | \alpha \rangle \langle \beta |, \hspace{8mm} \\
        L_j \rho L_j^\dag - \frac{1}{2} \left\{ L_j^\dag L_j, \rho \right\} \hspace{50mm} \\
        = \sum_{\alpha, \beta} \langle \alpha | \rho | \beta \rangle \left( c(j|\alpha) \overline{c(j|\beta)} - \frac{| c(j|\alpha) |^2}{2} - \frac{| c(j|\beta) |^2}{2} \right).
    \end{split} \end{equation*}
    Note that both expressions are zero when $\alpha = \beta$. So, the Lindbladian is $\mathcal{L} (\rho) = \sum_{\alpha \ne \beta} \langle \alpha | \rho | \beta \rangle \omega(\alpha,\beta) |\alpha\rangle \langle \beta |$ with
    \begin{equation*} \begin{split}
        \omega(\alpha,\beta) &= \imath \left( \epsilon(\alpha) - \epsilon(\beta) \right) \\
        & \hspace{5mm} + \sum_j \left( c(j|\alpha) \overline{c(j|\beta)} - \frac{1}{2} \theta(j|\alpha) - \frac{1}{2} \theta(j|\beta) \right).
    \end{split} \end{equation*}    
    
    If, under the non degeneracy condition, $\omega(\alpha,\beta)$ is non-zero, this will force any invariant state to be diagonal in the basis of the pointer states. Therefore, we will have proven the uniqueness of the decomposition. Indeed, otherwise, we should be able to find a non-diagonal invariant state according to Corollary \ref{Cor_Non_Unique_Minimal_Enclosures}. Given $\alpha \ne \beta$, denote
    %$z(\alpha,\beta) = \epsilon(\alpha) - \epsilon(\beta) \in \mathbb{R}$ and 
    $c(j|\alpha) = a_j^\alpha + \imath b_j^\alpha$ with $a_j^\alpha, b_j^\alpha \in \mathbb{R}$. This way, the real part of $\omega(\alpha,\beta)$ can be written as:
    % \begin{equation*} \begin{split}
    %     \Re{\left( \omega(\alpha,\beta) \right)} &= - \frac{1}{2} \sum_j \left( a_j(\alpha) - a_j(\beta) \right)^2 \\
    %     & \hspace{10mm} - \frac{1}{2} \sum_j \left( b_j(\alpha) - b_j(\beta) \right)^2
    % \end{split} \end{equation*}
    \begin{equation*}
        \Re{\left( \omega(\alpha,\beta) \right)} = - \frac{1}{2} \sum_j \left( \left( a_j^\alpha - a_j^\beta \right)^2 + \left( b_j^\alpha - b_j^\beta \right)^2 \right).
    \end{equation*}
    Noting that $r(j|\alpha) = a_j^\alpha$ and $\theta(j|\alpha) = (a_j^\alpha)^2 + (b_j^\alpha)^2$, under the non degeneracy condition, $\Re{\left( \omega(\alpha,\beta) \right)}$ is always negative for any $\alpha \ne \beta$, which finishes the proof.
\end{proof}

\hspace{2mm}

%It should be noted that we have found in the proof the explicit necessary and sufficient condition to have uniqueness of the decomposition of the Hilbert space for a quantum trajectory under nondemolition measurement. However, in the general case, finding the explicit condition is not always easy.
The proof provides an explicit necessary and sufficient condition for the uniqueness of the Hilbert space decomposition of a quantum trajectory under nondemolition measurement. However, in general, finding such a condition is not always easy.

\subsection{Continuous time without transient subspace}

We can now drop the nondemolition condition to consider more general cases. In \cite{amini2021asymptotic}, under an \textit{identifiability assumption}, it is shown that the quantum trajectory will select one of the minimal invariant subspaces at large time, assuming that there is no transient part. 

Denote $\mathcal{H} = \bigoplus_\alpha \mathcal{V}_\alpha$ an orthogonal decomposition of the Hilbert space into minimal enclosures and $\rho_{\mathcal{V}_\alpha}$ the associated extremal invariant states. The identifiability assumption is the following: for all $\alpha \ne \beta$, there is a $L_j$ such that
\begin{equation} \label{Eq_Identifiability_Assumption_Continuous}
    \tr \left( \left( L_j + L_j^\dag \right) \rho_{\mathcal{V}_\alpha} \right) \ne \tr \left( \left( L_j + L_j^\dag \right) \rho_{\mathcal{V}_\beta} \right).
\end{equation}

\begin{theorem}
    Assuming that the Hilbert space associated to the quantum trajectory has no transient part and that there exists a decomposition into minimal enclosures that satisfies the identifiability assumption (\ref{Eq_Identifiability_Assumption_Continuous}), then this decomposition is the only possible one.
\end{theorem}

\begin{proof}
    By contraposition, let us assume that the decomposition is not unique. Then, according to Corollary \ref{Cor_Non_Unique_Minimal_Enclosures} there exist $\alpha \ne \beta$ and a partial isometry $Q$ such that:
    \begin{itemize}
        \item $\rho_{\mathcal{V}_\alpha} = Q \rho_{\mathcal{V}_\beta} Q^\dag$,
        \item $Q^\dag Q = \mathbbm{1}_{\mathcal{V}_\beta}$,
        \item $Q$ is a fixed point of the cut off semigroup.
    \end{itemize}
    First, note that if we assume $\mathcal{D} = \{ 0 \}$, the cut off semigroup and the adjoint semigroup are the same. In this case, by Theorem 7.2 in \cite{wolf2012quantum}, $Q$ commutes with the $L_j$ appearing in the unraveling of the Lindbladian. Therefore,
    \begin{equation*} \begin{split}
        \tr \left( \left( L_j + L_j^\dag \right) \rho_{\mathcal{V}_\alpha} \right) &= \tr \left( \left( L_j + L_j^\dag \right) \rho_{\mathcal{V}_\beta} Q^\dag Q \right) \\
        &= \tr \left( \left( L_j + L_j^\dag \right) \rho_{\mathcal{V}_\beta} \right).
    \end{split} \end{equation*}
    So, the identifiability assumption cannot be satisfied.
\end{proof}

\subsection{Discrete time without transient subspace}

Until now, we have only considered continuous-time semigroups. However, the same decomposition also exists in discrete time, as demonstrated in \cite{carbone2016irreducible}, and we have the same results of convergence for quantum trajectories as well. In \cite{amini2024exponential}, assuming there is no transient part, the authors suggest an \textit{optimal identifiability assumption}, meaning that it is necessary and sufficient to have selection into one minimal enclosure as time goes to infinity, see the remark associated with \cite[Proposition 2.1]{amini2024exponential}. It takes the following form.

Consider a quantum channel $\Phi = \sum_j V_j \cdot V_j^\dag$ and an orthogonal decomposition of the associated Hilbert space $\mathcal{H} = \bigoplus_\alpha \mathcal{V}_\alpha$ with $\rho_{\mathcal{V}_\alpha}$ as before. Then, for all $\alpha \ne \beta$, there exist a sequence $I = \left( i_1, \dots, i_p \right)$ such that:
\begin{equation} \label{Eq_Identifiability_Assumption_Discrete}
    \tr \left( V_I \rho_{\mathcal{V}_\alpha} V_I^\dag \right) \ne \tr \left( V_I \rho_{\mathcal{V}_\beta} V_I^\dag \right),
\end{equation}
where $V_I = V_{i_p} \dots V_{i_1}$.

\begin{theorem}
    Assuming that the Hilbert space associated to the quantum trajectory has no transient part and that there exists a decomposition into minimal enclosures that satisfies the optimal identifiability assumption (\ref{Eq_Identifiability_Assumption_Discrete}), then this decomposition is the only possible one.
\end{theorem}

\begin{proof}
    The proof is the same as in the previous case. Simply replace Theorem 7.2 in \cite{wolf2012quantum} with Theorem 6.13 in \cite{wolf2012quantum} to match the discrete-time assumption.
\end{proof}

\vspace{2mm}

It is natural to wonder if this assumption is equivalent to the uniqueness of the decomposition of the Hilbert space. However, it is not the case, as shown in the following: take a rotation matrix $U \in \mathcal{M}_2 (\mathbb{C})$ different from the identity and define the quantum channel $\Phi := U \cdot U^\dag$. There exists $\theta \in ]0, 2 \pi [$ such that $U = \begin{pmatrix} \cos{\theta} & - \sin{\theta} \\ \sin{\theta} & \cos{\theta} \end{pmatrix}$. Let $\rho \in \mathcal{D}(\mathbbm{C}^2)$
%$\rho = \begin{pmatrix} a & b \\ \overline{b} & c \end{pmatrix}$
be a density matrix as in Section \ref{Section_App_Dim_2}. Then:
\begin{equation*} \begin{split}
    & \langle e_1 | \Phi (\rho) | e_1 \rangle = a \cos^2{\theta} - \left( b + \overline{b} \right) \cos{\theta} \sin{\theta} + c \cos^2{\theta}, \\
    & \langle e_1 | \Phi (\rho) | e_2 \rangle = (a - c) \cos{\theta} \sin{\theta} + b \cos^2{\theta} - \overline{b} \sin^2{\theta}, \\
    & \langle e_2 | \Phi (\rho) | e_1 \rangle = \overline{\langle e_1 | \Phi (\rho) | e_2 \rangle}, \\
    & \langle e_2 | \Phi (\rho) | e_2 \rangle = a \sin^2{\theta} + \left( b + \overline{b} \right) \cos{\theta} \sin{\theta} - c \cos^2{\theta}.
\end{split} \end{equation*}
After computation, the invariant states of $\Phi$ are of the form:
\begin{equation*}
    \rho_{inv} = \begin{pmatrix} 1/2 & \imath x \\ - \imath x & 1/2 \end{pmatrix}, \hspace{2mm} \mathrm{with} \hspace{2mm} x \in ]-1/2, 1/2[.
\end{equation*}
These are exactly the convex combinations of the two following extremal invariant states:
\begin{equation*}
    \rho_\alpha = | \psi_\alpha \rangle \langle \psi_\alpha | \hspace{2mm} \mathrm{and} \hspace{2mm} \rho_\beta = | \psi_\beta \rangle \langle \psi_\beta |,
\end{equation*}
with 
$| \psi_\alpha \rangle = \begin{pmatrix} 1 & \imath \end{pmatrix}^\top / \sqrt{2} \hspace{2mm} \mathrm{and} \hspace{2mm} | \psi_\beta \rangle = \begin{pmatrix} 1 & - \imath \end{pmatrix}^\top / \sqrt{2}$.

Therefore, the Hilbert space $\mathbb{C}^2$ admits only one decomposition: $\mathbb{C}^2 = \Span{ (| \psi_\alpha \rangle)} \oplus \Span{ (| \psi_\beta \rangle)}$. 
%\overset{\perp}{\oplus}
Still, 
\begin{equation*}
    \forall p \in \mathbb{N}^*, \ \tr \left( U^p \rho_\alpha \left( U^p \right)^\dag \right) = \tr \left( U^p \rho_\beta \left( U^p \right)^\dag \right) = 1.
\end{equation*}
%So, the identifiability assumption is not verified for this quantum channel.

%%%%%%%%%%%%%%%%%%%%%%%%%%%%%%%%%%%%%%%%%%%%%%%%%%%%%%%%%%%%%%%%%%%%%%%%%%%%%%%%

\section{CONCLUSIONS AND FUTURE WORKS}

In this paper, we propose a review of \cite{baumgartner2012structures} and \cite{carbone2016irreducible}. To make these results more accessible, we sometimes simplify them or provide new proofs. Moreover, we apply these results to minimal OQRW, confirming the robustness of the theory, and then to quantum trajectories, linking the assumptions needed for selection in the long-time limit to the uniqueness of the decomposition. In a future paper currently in progress, we investigate the generalization of the latter result to the case where the decomposition includes the transient part and additionally studying the selection of invariant subspaces and their feedback stabilization, extending the work of \cite{amini2024exponential,benoist2024exponentially} to continuous time. Another promising application of the decomposition is in the reservoir or dissipation engineering \cite{langbehn2024dilute,wolf2012quantum}, which is in our research lines.

%%%%%%%%%%%%%%%%%%%%%%%%%%%%%%%%%%%%%%%%%%%%%%%%%%%%%%%%%%%%%%%%%%%%%%%%%%%%%%%%
%\section{ACKNOWLEDGMENTS}
%Part of this research was conducted during a visit to the International Centre for Theoretical Sciences (ICTS) during ’Quantum Trajectories’ program.

%%%%%%%%%%%%%%%%%%%%%%%%%%%%%%%%%%%%%%%%%%%%%%%%%%%%%%%%%%%%%%%%%%%%%%%%%%%%%%%%

\bibliographystyle{plain}  
\bibliography{refs}

\end{document}